\pdfoutput=1
\RequirePackage{ifpdf}
\ifpdf 
\documentclass[pdftex]{sigma}
\else
\documentclass{sigma}
\fi

\usepackage{mathrsfs}

\newtheorem{Theorem}{Theorem}[section]
\newtheorem{Lemma}[Theorem]{Lemma}
\newtheorem{Proposition}[Theorem]{Proposition}
 { \theoremstyle{definition}
\newtheorem{Remark}[Theorem]{Remark} }

  \def\cF{\mathcal{F}}
  \def\cI{\mathcal{I}}

  \def\cU{\mathcal{U}}
 \def\cW{\mathcal{W}}

\newcommand{\QQ}{{\mathbb Q}}

\newcommand{\ZZ}{{\mathbb Z}}

\newcommand{\wh}[1]{\widehat{#1}}

\newcommand{\sfrac}[2]{{\textstyle{\frac{#1}{#2}}}}
\newcommand{\gelpn}{{\mathcal{A}_{q,p}\big(\widehat{\mathfrak{gl}}(N)_{c}\big)}}
\newcommand{\gelp}{{\mathcal{A}_{q,p}\big(\widehat{\mathfrak{gl}}(2)_{c}\big)}}

\newcommand{\ellipt}[1]{\mbox{\AA${}_{q,p,c}\big(\wh{\mathfrak{gl}}_{#1}\big)$}}

\numberwithin{equation}{section}

\begin{document}
\allowdisplaybreaks

\newcommand{\arXivNumber}{2005.03579}

\renewcommand{\thefootnote}{}

\renewcommand{\PaperNumber}{094}

\FirstPageHeading

\ShortArticleName{On Abelianity Lines in Elliptic $W$-Algebras}

\ArticleName{On Abelianity Lines in Elliptic $\boldsymbol{W}$-Algebras\footnote{This paper is a~contribution to the Special Issue on Elliptic Integrable Systems, Special Functions and Quantum Field Theory. The full collection is available at \href{https://www.emis.de/journals/SIGMA/elliptic-integrable-systems.html}{https://www.emis.de/journals/SIGMA/elliptic-integrable-systems.html}}}

\Author{Jean AVAN~$^\dag$, Luc FRAPPAT~$^\ddag$ and Eric RAGOUCY~$^\ddag$}

\AuthorNameForHeading{J.~Avan, L.~Frappat and E.~Ragoucy}

\Address{$^\dag$~Laboratoire de Physique Th\'eorique et Mod\'elisation,
CY Cergy Paris Universit\'e, \\
\hphantom{$^\dag$}~CNRS, F-95302 Cergy-Pontoise, France}
\EmailD{\href{mailto:avan@u-cergy.fr}{avan@u-cergy.fr}}

\Address{$^\ddag$ Laboratoire d'Annecy-le-Vieux de Physique Th{\'e}orique LAPTh, Universit\'e Grenoble Alpes,\\
\hphantom{$^\ddag$}~USMB, CNRS, F-74000 Annecy, France}
\EmailD{\href{mailto:luc.frappat@lapth.cnrs.fr}{luc.frappat@lapth.cnrs.fr}, \href{mailto:eric.ragoucy@lapth.cnrs.fr}{eric.ragoucy@lapth.cnrs.fr}}

\ArticleDates{Received May 08, 2020, in final form September 22, 2020; Published online September 30, 2020}

\Abstract{We present a systematic derivation of the abelianity conditions for the $q$-deformed $W$-algebras constructed from the elliptic quantum algebra $\mathcal{A}_{q,p}\big(\widehat{\mathfrak{gl}}(N)_{c}\big)$. We identify two sets of conditions on a given critical surface yielding abelianity lines in the moduli space ($p, q, c$). Each line is identified as an intersection of a countable number of critical surfaces obeying diophantine consistency conditions. The corresponding Poisson brackets structures are then computed for which some universal features are described.}


\Classification{17B37; 17B68}

\renewcommand{\thefootnote}{\arabic{footnote}}
\setcounter{footnote}{0}

\section{Introduction}\label{sect:intro}
The construction of deformed $W_N$ algebras as subalgebras of the elliptic quantum algebra $\mathcal{A}_{q,p}\big(\widehat{\mathfrak{gl}}(N)_{c}
\big)$ was proposed in \cite{AFR19}. The construction uses
as generating functionals quadratic and higher rank traces of the quantum Lax operators defining
$\mathcal{A}_{q,p}\big(\widehat{\mathfrak{gl}}(N)_{c}\big)$.
The existence of such closed subalgebras of the enveloping algebra of $\mathcal{A}_{q,p}\big(\widehat{\mathfrak{gl}}(N)_{c}\big)$ was conditioned by a so-called ``critical'' relation between the elliptic modulus or nome~$p$, the quantum deformation parameter~$q$ and the central charge~$c$. This critical relation is parametrized by two integers~$(m,n)$, and defines surfaces $\mathscr{S}_{m,n}$ in the $(p,q,c)$ moduli space.
The structure functions were identified as particular ratios of elliptic functions.
Characterizing these structures as $q$-deformed $W$-algebras was made possible by first finding a second constraint on $p$, $q$, $c$, yielding now structure functions degenerating to~$1$.
This second constraint may thus be consistently called ``abelianity condition''
and defines a line on the surface $\mathscr{S}_{m,n}$.
The expansion of the structure functions around this constraint, by infinitesimally relaxing it, yields Poisson structures, which could then be compared to, and in some cases identified with, the original ones in \cite{FR1, FR}.
The full quantum structures could then be identified as natural quantizations of these Poisson structures.
The derivation of this second ``abelianity constraint'' however assumed a very specific pattern of cancellation inside the elliptic structure functions.
It was therefore a natural question whether more general cancellation patterns occur which may then lead to new abelianity conditions (and as a consequence new Poisson structures).
We will address this issue here, and determine the most general cancellation pattern of the structure functions, within a given ``fundamental'' scheme using the particular form of the structure functions as ratios of products of a single elliptic function with shifted/modified arguments, and a remarkable periodicity property of this component function.
Quite remarkably, it turns out that:
\begin{enumerate}\itemsep=0pt
\item All abelianity lines are identified as intersections of critical surfaces (generically a countable set of such surfaces).
\item Intersections of critical surfaces yielding abelianity lines are characterized by a diophantine-type condition of integrity of a certain ratio of combinations of their integer parameters.
\end{enumerate}

We shall now detail this derivation, starting with a reminder of the general frame of~\cite{AFR19} and prepare some notations. The main result shall then be stated precisely, and its proof will be given in a detailed way.
We then compute explicit associated Poisson structures and compare them along different surfaces converging onto the same line.
These Poisson structures are realized by linear combinations of a few fundamental elliptic functions, identified as logarithmic derivatives of the short Jacobi theta function. The specific elliptic functions depend only on the abelianity line itself, whichever realization by an intersection is achieved.
Only the constant rational coefficients and the span of the linear sums depend on the critical surface along which the PB structure is expanded.

\section[Quadratic subalgebras in $\gelpn$]{Quadratic subalgebras in $\boldsymbol{\gelpn}$}\label{sect:Wpq}

The central object of our study are quadratic subalgebras $\cW^{(m,n)}_{pqc}(N)$ in the
quantum elliptic algebra $\gelpn$, parametrized by two integers $m,n\in\ZZ$ (in addition to the parameters~$p$,~$q$,~$c$ of the quantum elliptic algebra). We refer to \cite{AFR17,AFR19} for the full construction, and summarize the main points needed for our present study.

The $\cW^{(m,n)}_{pqc}(N)$ subalgebras are defined on surfaces in the three-dimensional parameter space spanned by $(q,p,c)$
\begin{equation}\label{eq:surf}
\mathscr{S}_{m,n}\colon \  \big({-}p^{\frac{1}{2}}\big)^{m} \big({-}p^{*\frac{1}{2}}\big)^{n} = q^{-N} ,
\end{equation}
where $p^*=pq^{-2c}$. We introduce the operators $t_{m,n}^{(k)}(z)$, $1\leq k\leq N$, generating the subalgebra
$\cW^{(m,n)}_{pqc}(N)$, and $L(z)$ the Lax operator of the $\gelpn$ algebra.
The $\cW^{(m,n)}_{pqc}(N)$ algebra is defined by the following proposition, proved in~\cite{AFR19}:
\begin{Proposition}
On the surface $\mathscr{S}_{m,n}$, one has:
\begin{enumerate}\itemsep=0pt
\item[$a)$] The generators $t_{m,n}^{(k)}(z)$ obey the following exchange relation with $L(w)$:
\begin{equation}
t_{m,n}^{(k)}(z) L(w) = \prod_{i=1}^{k}\frac{\cF_{-m}(z_i/w) }{ \cF^*_{n}(z_i/w)} L(w) t_{m,n}^{(k)}(z) .
\label{eq:exchtL}
\end{equation}
The function $\cF_a(x)$ is expressed in terms of the function $\cU(x)$ defined in \eqref{eq:defU} as
\begin{equation*}
\cF_{a}(x) =
\begin{cases}
\displaystyle \prod_{\ell=0}^{a-1} \cU\big(\big({-}p^{\frac{1}{2}}\big)^\ell x\big) & \text{for $a > 0$},
\\
1 & \text{for $a = 0$}, \\
\displaystyle \prod_{\ell=1}^{|a|} \cU\big(\big({-}p^{\frac{1}{2}}\big)^{-\ell} x\big)^{-1} & \text{for $a < 0$},
\end{cases}
\qquad \cF^*_a(x) = \cF_a(x)\big\vert_{p \to p^*}. 
\end{equation*}

\item[$b)$] They realize quadratic subalgebras in $\gelpn$ with quadratic exchange relations for $1 \le k,k' \le N$:
\begin{equation}\label{eq:exchtt}
t_{m,n}^{(k)}(z) t_{m,n}^{(k')}(w) = \prod_{i=(1-k)/2}^{(k-1)/2} \prod_{j=(1-k')/2}^{(k'-1)/2} \mathcal{Y}_{m,n}\big(q^{i-j}z/w\big) t_{m,n}^{(k')}(w) t_{m,n}^{(k)}(z),
\end{equation}
where the function $\mathcal{Y}_{m,n}(x)$ is given by
\begin{equation}\label{eq:funcY}
\mathcal{Y}_{m,n}(x) = \frac{\cF_{n}^*(x)\cF_{-n}^*(x)}{\cF_{m}(x)\cF_{-m}(x)}
= \frac{\displaystyle \prod_{\ell=1}^{|m|} \mathcal{U}\big(\big({-}p^{\frac{1}{2}}\big)^{-\ell} x\big)
\prod_{\ell=1}^{|n|-1} \mathcal{U}\big(\big({-}p^{*\frac{1}{2}}\big)^{\ell} x\big)}
{\displaystyle \prod_{\ell=1}^{|m|-1} \mathcal{U}\big(\big({-}p^{\frac{1}{2}}\big)^{\ell} x\big)
\prod_{\ell=1}^{|n|} \mathcal{U}\big(\big({-}p^{*\frac{1}{2}}\big)^{-\ell} x\big)} .
\end{equation}
The function $\cU(z)$ is defined using the short Jacobi $\theta$ function
\begin{equation}\label{eq:defU}
\cU(z)=q^{\frac2N-2} \frac{\theta_{q^{2N}}\big(q^2z^2\big) \theta_{q^{2N}}\big(q^2z^{-2}\big)} {\theta_{q^{2N}}\big(z^2\big)\theta_{q^{2N}}\big(z^{-2}\big)} .
\end{equation}
\end{enumerate}
\end{Proposition}
We remind that the short Jacobi $\theta$ function is defined in terms of the infinite $q$-Pochhammer symbols $
(z;a)_\infty = \prod\limits_{n \ge 0} (1-z a^{n})$ by
\begin{equation*}
\theta_a(z) = (z;a)_\infty \big(az^{-1};a\big)_\infty .
\end{equation*}
It enjoys the following properties
\begin{equation*}
\theta_{a^2}\big(a^2z\big) = \theta_{a^2}\big(z^{-1}\big) = -\frac{\theta_{a^2}(z)}{z}\qquad
\text{and} \qquad \theta_{a^2}(az) = \theta_{a^2}\big(az^{-1}\big) .
\end{equation*}
To ensure a proper definition of the elliptic quantum algebra $\gelpn$, and in particular the convergence of the short Jacobi $\theta$ functions as infinite products, we have to suppose that $|p|<1$ and $|q|<1$.

\begin{Remark}\label{rmk:22}
At this point, we need to elaborate the exact meaning of the exchange relation~\eqref{eq:exchtt}. It is to be understood as an equality of formal series expansion after a suitable Riemann--Hilbert splitting of the meromorphic function $\mathcal{Y}_{m,n}(x)$. It will then acquire supplementary terms (central extensions or higher spin operators) on poles or zeroes of this exchange function. We do not achieve this procedure here since we are only interested in the abelianity conditions, see below.
\end{Remark}

{\bf Line of abelianity.} We are interested in characterizing the situations where the subalgebra~\eqref{eq:exchtt} in $\gelpn$ becomes abelian. Demanding an abelian exchange relation between the generators $t_{m,n}^{(1)}(z)$ imposes $\mathcal{Y}_{m,n}(x) = 1$.
Once abelianity is obtained between these generators, the whole subalgebra \eqref{eq:exchtt} becomes abelian, because the exchange function in~\eqref{eq:exchtt} for generic $k$, $k'$ is a product of the $\mathcal{Y}_{m,n}$ function with shifted arguments. The conditions allowing $\mathcal{Y}_{m,n}(x) = 1$ will define a \emph{line of abelianity}.

Since this strict abelian condition implies that $\mathcal{Y}_{m,n}(x)$ has neither a pole nor a zero, the extra terms mentioned in Remark \ref{rmk:22} are not expected to appear and abelianity is indeed exact.
They will nevertheless contribute to the Poisson brackets computed in the neighborhood of the abelianity line, see Section~\ref{sect:poisson}.

In the following, it will be convenient to parametrize $p$ and $c$ (or equivalently $p$ and $p^*$)
on the surface $\mathscr{S}_{m,n}$ through the relation
\begin{equation}\label{eq:paramppstar}
-p^{\frac{1}{2}} = q^{-N\lambda/m} \qquad \text{and} \qquad -p^{*\frac{1}{2}} = q^{-N\lambda^*/n}  .
\end{equation}
The surface condition then reads $\lambda+\lambda^*=1$, and the central charge is given by
\begin{equation}\label{eq:cc}
c = -N \left( \frac{\lambda}{m}-\frac{\lambda^*}{n} \right) .
\end{equation}

A first step in deriving abelianity conditions was performed in \cite{AFR17,AFR19}, where the following result was shown:
\begin{Proposition}\label{prop:ligne}
On the surface $\mathscr{S}_{m,n}$, the generators $t_{m,n}^{(k)}(z)$ realize an abelian subalgebra in $\mathcal{A}_{q,p}\big(\widehat{\mathfrak{gl}}(N)_{c}\big)$ when $\lambda$, $\lambda^*$ take non-vanishing \emph{integer} values.
\end{Proposition}
Note that the notation $\lambda^*$ corresponds to $\lambda'$ in \cite{AFR17,AFR19} (it is \emph{not} the complex conjugate of~$\lambda$!).

\section{Main results\label{sect:main}}
In this section, we expose our main results concerning the classification of the lines of abelianity, and their realization as intersections of critical surfaces.
\subsection{Generic abelianity lines as intersections}

\begin{Lemma}\label{lem:det}
Two different surfaces $\mathscr{S}_{m,n}$ and $\mathscr{S}_{m',n'}$ have a non-empty intersection if and only if $m\ne m'$, $n\ne n'$ and the determinant $\left|\begin{smallmatrix} m & m' \\ n & n' \end{smallmatrix}\right| \ne 0$.
 In that case, the parameters $p$, $p^*$ and $c$ are given by
\begin{equation}\label{eq:ppstar}
-p^{\frac{1}{2}} = q^{N\sfrac{n'-n}{m'n-mn'}} , \qquad
-p^{*\frac{1}{2}} = q^{N\sfrac{m-m'}{m'n-mn'}} ,
\end{equation}
and
\begin{equation}\label{eq:central}
c = N \frac{m'+n'-m-n}{m'n-mn'} .
\end{equation}
\end{Lemma}
\begin{proof}
If two surfaces $\mathscr{S}_{m,n}$ and $\mathscr{S}_{m',n'}$ intersect, then $p$, $q$, $c$ have to satisfy simultaneously the two surface conditions $\big({-}p^{\frac{1}{2}}\big)^{m} \big({-}p^{*\frac{1}{2}}\big)^{n} = q^{-N}$ and $\big({-}p^{\frac{1}{2}}\big)^{m'} \big({-}p^{*\frac{1}{2}}\big)^{n'} = q^{-N}$.
This implies $\big({-}p^{\frac{1}{2}}\big)^{m'-m} = \big({-}p^{*\frac{1}{2}}\big)^{n-n'}$.
Thus $n=n'$ implies that $m=m'$ (and vice-versa). But then the
two surfaces coincide, which contradicts the hypothesis.
Hence, we must have $m\neq m'$ and $n\neq n'$.
Given that $p^*=pq^{-2c}$, the surface conditions can be rewritten
 as $\big({-}p^{\frac{1}{2}}\big)^{m+n} = q^{nc-N}$
and $\big({-}p^{\frac{1}{2}}\big)^{m'+n'} = q^{n'c-N}$, which leads to
\begin{eqnarray}
&&\label{eq:compat}
\big({-}p^{\frac{1}{2}}\big)^{m'n-mn'} = q^{N(n'-n)} \qquad \text{and} \qquad
\big({-}p^{\frac{1}{2}}\big)^{m'+n'-m-n} = q^{c(n'-n)} .
\end{eqnarray}
Then, since $q$ is not a root of unity, the first equation in \eqref{eq:compat} has a solution if and only if $n'm-m'n \ne 0$.
In that case, \eqref{eq:ppstar} and \eqref{eq:central} follow immediately.
Note that these relations are invariant in the exchange $(m,n) \leftrightarrow (m',n')$.
\end{proof}

\begin{Remark} Recalling the parametrization \eqref{eq:paramppstar}, one sees that on the intersection of two surfaces
$\mathscr{S}_{m,n}$ and $\mathscr{S}_{m',n'}$, one has the following values for the line viewed on $\mathscr{S}_{m,n}$
\begin{equation}\label{eq:lambda-inter}
\lambda=\frac{m(n-n')}{m'n-mn'}
\qquad \text{and}\qquad
\lambda^*=\frac{n(m'-m)}{m'n-mn'} .
\end{equation}
Note that the surface condition $\lambda+\lambda^*=1$ is then automatically satisfied.
\end{Remark}

\begin{Lemma}\label{lem:intersect}
Let $\mathscr{S}_{m,n}$ and $\mathscr{S}_{m',n'}$ be two surfaces with a non empty
intersection, defining a~line in the moduli space. There are a countable number of
surfaces $\mathscr{S}_{m'',n''}$ intersecting on this line. They are uniquely determined by the relation
\begin{equation}\label{eq:trois}
\frac{m'-m}{n'-n} = \frac{m''-m'}{n''-n'} .
\end{equation}
\end{Lemma}
\begin{proof}
Suppose there is a third surface $\mathscr{S}_{m'',n''}$ intersecting on
$\mathscr{S}_{m,n}\cap\mathscr{S}_{m',n'}$. Then equations~\eqref{eq:ppstar} and \eqref{eq:central} hold for any choice of two pairs in $\{(m,n),(m',n'),(m'',n'')\}$:
\begin{gather}
\label{eq:troisA} \frac{n'-n}{m'n-mn'} = \frac{n''-n'}{m''n'-m'n''} = \frac{n''-n}{m''n-mn''} , \\
\label{eq:troisB} \frac{m'+n'-m-n}{m'n-mn'} = \frac{m''+n''-m'-n'}{m''n'-m'n''} = \frac{m''+n''-m-n}{m''n-mn''} .
\end{gather}
Dividing term by term the second equation by the first one, one gets
\begin{equation}\label{eq:troisC}
\frac{m'-m}{n'-n} = \frac{m''-m'}{n''-n'} = \frac{m''-m}{n''-n} .
\end{equation}
Thus, relation \eqref{eq:trois} is a necessary relation for the surface $\mathscr{S}_{m'',n''}$ to exist. Note that this relation implies the second equality in \eqref{eq:troisC}, i.e., assuming \eqref{eq:trois} implies all the equalities deduced by circular permutation of the pairs $(m,n)$, $(m',n')$, $(m'',n'')$.

Now writing \eqref{eq:trois} as $(m'-m)(n''-n')-(m''-m')(n'-n)=0$ and multiplying it by~$n'$, one finds the first equality of \eqref{eq:troisA}. Multiplying instead by~$m'$, one finds the first equality of~\eqref{eq:troisB}. The other two equalities are found by a circular permutation on the pairs $(m,n)$, $(m',n')$, $(m'',n'')$. Hence equation~\eqref{eq:trois} is equivalent to \eqref{eq:troisA}--\eqref{eq:troisB}. As such, it is a necessary and sufficient condition for $\mathscr{S}_{m'',n''}$ to exist.

Finally, choosing $m''=m' + u(m-m')$ and $n''=n'+u(n-n')$ with $u\in\ZZ$, one shows at the same time that~\eqref{eq:trois} admits solutions, and that there are a countable number of them.
\end{proof}

\begin{Theorem}\label{thm:line}
When non-empty, the intersection of two surfaces $\mathscr{S}_{m,n}$ and $\mathscr{S}_{m',n'}$ is a line of abelianity of $\mathscr{S}_{m,n}$ if and only if one of the following conditions is satisfied:
\begin{gather}
(a) \quad \frac{m(n-n')}{m'n-mn'} \in \ZZ \qquad \text{or equivalently} \qquad \frac{n(m'-m)}{m'n-mn'} \in \ZZ, \label{eq:thmb}\\
(b) \quad \frac{m+n-m'-n'}{m'n-mn'} \in \ZZ \qquad\text{and} \qquad
(m+n)(m'+n')\neq0,\label{eq:thma}\\
(c)\quad (m',n')=\pm(1,-1) \qquad \text{and} \qquad m,n\in\ZZ,\nonumber\\
(c')\quad (m,n)=\pm(1,-1) \qquad \text{and}\qquad m',n'\in\ZZ .\nonumber
\end{gather}
In case $(a)$, the intersection might not be a line of abelianity of $\mathscr{S}_{m',n'}$. In cases~$(b)$,~$(c)$ and~$(c')$, the intersection is also a~line of abelianity of~$\mathscr{S}_{m',n'}$.
\end{Theorem}

\begin{Theorem}\label{thm:surf}Any line of abelianity on a surface $\mathscr{S}_{m,n}$ can be identified with the intersection of a~countable number of suitable surfaces $\mathscr{S}_{m',n'}$.
\end{Theorem}

The proof of these two theorems is postponed until Section~\ref{sect:proofs}.

\subsection{Enhanced abelianities}

Theorems~\ref{thm:line} and~\ref{thm:surf} describe generic lines of abelianity on critical surfaces $\mathscr{S}_{m,n}$, i.e., such that the sole characterizing commutation property be $\big[t^{(k)}_{m,n}(z), t^{(k')}_{m,n}(w)\big] = 0$, $1\leq k,k'\leq N$. We have identified several critical surfaces on which stronger commutation properties prevail, which can be overall characterized as ``enhanced abelianity''. They are described in the following propositions.

\begin{Proposition}\label{prop:Scrit}\quad
\begin{enumerate}\itemsep=0pt
\item[$(i)$]
The surface $\mathscr{S}_{1,-1}$ corresponds to the critical level $c\!=\!-N$, for which the genera\-tors~$t^{(k)}_{1,-1}(z)$ lie in the extended center of $\gelpn$, without any further condition on $q$, $p$.
\item[$(ii)$] The intersection of $\mathscr{S}_{1,-1}$ with any surface
 $\mathscr{S}_{m',n'}$, when non-empty, provides an abelianity line for $\mathscr{S}_{m',n'}$. The set of such lines is dense on the surface $\mathscr{S}_{1,-1}$.
\end{enumerate}
\end{Proposition}

\begin{Proposition}\label{prop:S0n}\quad
\begin{enumerate}\itemsep=0pt
\item[$(i)$] The generators $t_{0,n}^{(k)}(z)$ satisfy an abelian algebra on the whole surface $\mathscr{S}_{0,n}$.
\item[$(ii)$] There are a countable number of surfaces $\mathscr{S}_{m',n'}$ such that the intersection of $\mathscr{S}_{m',n'}$ with $\mathscr{S}_{0,n}$ is an abelianity line for $\mathscr{S}_{m',n'}$. These lines of abelianity form a dense set of lines on the surface $\mathscr{S}_{0,n}$.
\item[$(iii)$] There exist a countable number of surfaces $\mathscr{S}_{m',n'}$ such that the intersection of $\mathscr{S}_{m',n'}$ with~$\mathscr{S}_{0,n}$ is not an abelianity line for $\mathscr{S}_{m',n'}$.
\end{enumerate}
\end{Proposition}

\begin{Proposition}\label{prop:Smmm} \quad
\begin{enumerate}\itemsep=0pt
\item[$(i)$] The generators $t_{m,-m}^{(k)}(z)$ commute with the Lax operator $L(w)$
 of the elliptic quantum algebra $\gelpn$ on the surface $\mathscr{S}_{m,-m}$ $($``localized extended center''$)$ when
 the following conditions hold, see~\eqref{eq:paramppstar}:
\begin{enumerate}\itemsep=0pt
\item[$1)$] $m$ is odd;
\item[$2)$] $m$ and $\lambda$ are coprime integers;
\item[$3)$] $m$ and $\beta'_0+1$ are coprime integers, where $\beta_0$ and $\beta'_0$ are the B\'ezout coefficients such that $\beta_0 m - \beta'_0 \lambda=1$ with $1 \le \beta'_0 \le m-1$.
\end{enumerate}
Such a submanifold of the surface $\mathscr{S}_{m,-m}$ will be called a ``super-abelianity line''.
\item[$(ii)$] The intersection $\mathscr{S}_{m,-m} \cap \mathscr{S}_{1,1}$ with $m$ odd is a super-abelianity line for $\mathscr{S}_{m,-m}$ but it is never an abelianity line for $\mathscr{S}_{1,1}$.
\end{enumerate}
\end{Proposition}

The proof of these propositions is postponed to Section~\ref{sect:proofs}.

\begin{Remark}
The construction can be extended to the case of the algebra $\ellipt{N}$ defined with the unitary elliptic $R$-matrix \cite{AFR17,AFR19}. However the fact that the $R$-matrix is normalized in a different way modifies the analysis of the abelianity conditions. While it is easy to see that condition~(a) of Theorem~\ref{thm:line} still applies, the analysis of the other conditions is much more involved and remains to be done. It would be interesting since for~$\ellipt{N}$ the surface $\mathscr{S}_{2,-1}$ corresponds precisely to the original $q$-deformed $W$-algebras introduced in \cite{FR1, FR}.
\end{Remark}

\section{Abelianity lines}

We first introduce the following lemma which classifies the different lines of abelianity. This lemma is essential in proving the results exposed in Section~\ref{sect:main}.
\begin{Lemma}\label{lem:abel}
On the surface $\mathscr{S}_{m,n}$, the elliptic nomes $p$ and $p^*$ being parametrized as in~\eqref{eq:paramppstar} and the central charge given by \eqref{eq:cc},
the generators $t_{m,n}^{(k)}(z)$ realize an abelian subalgebra in $\mathcal{A}_{q,p}\big(\widehat{\mathfrak{gl}}(N)_{c}\big)$ if one of the following conditions is satisfied:{\samepage
\begin{enumerate}\itemsep=0pt
\item[$1.$] The parameters $\lambda$ and $\lambda^*$ take integer values.
\item[$2.$] The parameters $\lambda$ and $\lambda^*$ obey the relations
 $\frac{\lambda}{m}-\frac{\lambda^*}{n} \in \ZZ$ and $\frac{m+n}d\in\ZZ$, where $d$ is the denominator of the irreducible fraction of $\frac{\lambda}{m}$ $($or equivalently $\frac{\lambda^*}{n})$.
\end{enumerate}
When $N>2$, these conditions are necessary and sufficient conditions.}
\end{Lemma}

The proof of this lemma is given in Section~\ref{proof:abel}. Remark that the first part of condition 2 amounts to say that the central charge is an integer proportional to the critical value $-N$, see equation~\eqref{eq:cc}.
For completeness, we solved the condition $(2)$ of Lemma~\ref{lem:abel}:
\begin{Lemma}\label{lem:abel2}
Let $\ell$ and $\ell'$ be B\'ezout coefficients satisfying $\ell m + \ell' n = g$, where $g=\gcd(m,n)$.
Then, the parameters $\lambda$ and $\lambda^*$ solutions to the condition $(2)$ of Lemma~{\rm \ref{lem:abel}} are given by
\begin{equation*}
\frac{\lambda}{m} = \gamma'\ell + \frac{\gamma}{d} + \frac{kn}{g}
\qquad \text{and} \qquad
\frac{\lambda^*}{n} = \gamma'\ell' + \frac{\gamma}{d} - \frac{km}{g} ,\qquad k\in\ZZ,
\end{equation*}
where $d$ is a divisor of $m+n$ such that $\frac{m+n}d$ is coprime with $g$, and
$(\gamma',\gamma)$ are the B\'ezout coefficients solution of $\gamma' g + \gamma \frac{m+n}{d} = 1$. The integer $\gamma$ has to be coprime with $d$ and such that $0 < \gamma < d$.
\end{Lemma}

The proof of this lemma is given in Section~\ref{proof:abel2}.

\subsection{Proof of Lemma \ref{lem:abel}\label{proof:abel}}

Before going into the details, let us stress that we restrict ourselves to a framework where
cancellations between the numerator and the denominator of the exchange function
$\mathcal{Y}_{m,n}$ are done through the functions $\cU$ as a whole. Indeed, due to the explicit form of $\cU$, it seems very hard to obtain cancellation in a different way, even when dealing with some ``magic'' simplifications among elliptic functions. Remark however that for $N=2$, these 'magic' simplifications may occur, because the shift~$q^2$ in the definition of $\cU$, see~\eqref{eq:defU}, coincide with $q^N$, the half-period of the $\theta$ functions. Thus, the proof done here is only a sufficient condition when $N=2$.

We start with the expression \eqref{eq:funcY} of the function $\mathcal{Y}_{m,n}$ with the parametrization \eqref{eq:paramppstar}, where~$\lambda$,~$\lambda^*$ at this stage are complex numbers. The conditions for which $\mathcal{Y}_{m,n}(x)=1$ are determined by looking at the different ways the functions $\cU$ entering in~\eqref{eq:funcY} may simplify each other.
The simplification of the $\cU$ functions can be done essentially in two distinct ways:
\begin{enumerate}\itemsep=0pt
\item[(1)]
The functions $\mathcal{U}\big(\big({-}p^{\frac{1}{2}}\big)^{\ell} x\big)$ in the numerator cancel the functions $\mathcal{U}\big(\big({-}p^{\frac{1}{2}}\big)^{\ell'} x\big)$ in the denominator, and similarly for the functions $\mathcal{U}\big(\big({-}p^{*\frac{1}{2}}\big)^{\ell} x\big)$. In other words, one assumes that no ``cross-cancellations'' occur between $p$ and $p^*$ shifted-terms.
This case corresponds to the study done in \cite{AFR17,AFR19} and reminded in Proposition~\ref{prop:ligne}.
\item[(2)]
There exist at least one pair $(\ell,\ell')$ of indices such that the function $\mathcal{U}\big(\big({-}p^{\frac{1}{2}}\big)^{\ell} x\big)$ simplifies with the function $\mathcal{U}\big(\big({-}p^{*\frac{1}{2}}\big)^{\ell'} x\big)$. Taking into account the $q^N$-periodicity of the func\-tion~$\cU(x)$, this leads to
\begin{equation*}
\frac{\lambda\ell}{m} - \frac{\lambda^*\ell'}{n} \in \ZZ.
\end{equation*}
Since the surface condition reads $\lambda+\lambda^*=1$, this last equation implies $\lambda,\lambda^*\in\QQ$.
\end{enumerate}
We focus here on case (2), since case (1) was dealt with in \cite{AFR17}.
The parameters $\lambda$, $\lambda^*$ being then rational numbers, one sets
\begin{equation*}
\frac{\lambda}{m}=\frac{a}{d}\qquad \text{and}\qquad \frac{\lambda^*}{n}=\frac{b}{d'} ,
\end{equation*}
 where $(a,d)$ and $(b,d')$ are pairs of coprime numbers with $d,d'>0$.

\textbf{Step 1:} Writing $|m|=ds+\mu$ and $|n|=d's'+\mu'$ with $0<\mu<d$ and $0<\mu'<d'$, allows one to obtain a first simplification of the function $\mathcal{Y}_{m,n}(x)$ (note that the values $\mu=0$, $\mu'=0$, lead to $\lambda,\lambda^*\in\ZZ$ which returns to case~1).
Using the $q^N$-periodicity of the function $\cU$, one gets
\begin{equation*}
\frac{\displaystyle \prod_{\ell=1}^{|m|} \mathcal{U}\big(\big({-}p^{\frac{1}{2}}\big)^{-\ell} x\big)}
{\displaystyle \prod_{\ell=1}^{|m|-1} \mathcal{U}\big(\big({-}p^{\frac{1}{2}}\big)^{\ell} x\big)} =
\frac{\displaystyle \prod_{j=1}^{\mu} \mathcal{U}\big(q^{Naj/d} x\big)}
{\displaystyle \prod_{j'=d-\mu+1}^{d-1} \mathcal{U}\big(q^{Naj'/d} x\big)} .
\end{equation*}
Depending whether $\mu \le d-\mu$ or $\mu \ge d-\mu$, additional simplications may occur. In any case, introducing $\bar\mu = \inf(\mu,d-\mu) \le d/2$, one checks that
\begin{equation}\label{eq:Ymnreduc}
\frac{\displaystyle \prod_{j=1}^{\mu} \mathcal{U}\big(q^{Naj/d} x\big)}
{\displaystyle \prod_{j'=d-\mu+1}^{d-1} \mathcal{U}\big(q^{Naj'/d} x\big)} =
\frac{\displaystyle \prod_{j=1}^{\bar\mu} \mathcal{U}\big(q^{Naj/d} x\big)}
{\displaystyle \prod_{j'=d-\bar\mu+1}^{d-1} \mathcal{U}\big(q^{Naj'/d} x\big)} .
\end{equation}
Noting that the maximal value of $|j-j'|$ is $d-2$ and the minimal value is 1 in the r.h.s.\ of equation \eqref{eq:Ymnreduc}, $a(j-j')/d$ cannot be an integer, hence no further simplification occurs in the r.h.s.\ of equation \eqref{eq:Ymnreduc}.

The product of the $\mathcal{U}\big(\big({-}p^{*\frac{1}{2}}\big)^{-\ell} x\big)$ functions is processed in a similar way. Therefore, one obtains
\begin{equation}\label{eq:Ymnreductot}
\mathcal{Y}_{m,n}(x) = \frac{\displaystyle \prod_{j=1}^{\bar\mu} \mathcal{U}\big(q^{Naj/d} x\big)}
{\displaystyle \prod_{j=d-\bar\mu+1}^{d-1} \mathcal{U}\big(q^{Naj/d} x\big)}
\frac{\displaystyle \prod_{j'=d-\bar\mu'+1}^{d-1} \mathcal{U}\big(q^{Nbj'/d'} x\big)}
{\displaystyle \prod_{j'=1}^{\bar\mu'} \mathcal{U}\big(q^{Nbj'/d'} x\big)} .
\end{equation}
No further simplification can occur ``vertically'', hence it is necessary, for the function $\mathcal{Y}_{m,n}(x)$ to be equal to one, that the product $\prod\limits_{j=1}^{\bar\mu} \mathcal{U}\big(q^{Naj/d} x\big)$ simplifies the product $\prod\limits_{j'=1}^{\bar\mu'} \mathcal{U}\big(q^{Nbj'/d'} x\big)$, which implies $\bar\mu = \bar\mu'$.

Let $d=\alpha\delta$, $d'=\alpha\delta'$ where $\delta,\delta'$ are coprimes numbers. The term $\mathcal{U}\big(q^{Na/d} x\big)$ has to simplify some term $\mathcal{U}\big(q^{Nbk/d'} x\big)$, one can write $z=a/d-bk/d' \in \ZZ$ for a certain $1 \le k \le \bar\mu$, hence $a\delta'=\delta(bk+\alpha z\delta')$.
This last equation implies $\delta=1$ since $(a,\delta)$ and $(\delta',\delta)$ are pairs of coprimes numbers. Similarly, $z'=b/d'-ak'/d \in \ZZ$ for a certain $1 \le k' \le \bar\mu$, which leads to $\delta'=1$. One concludes that $d=d'$.

\textbf{Step 2:} We look now at the possible cross-simplifications in \eqref{eq:Ymnreductot} (``matching condition'').
A~simplification occurs whenever the argument of a $\cU$ function in the upper left product matches the argument of a $\cU$ function in the lower right product up to a power of $q^N$ due to the $q^N$-periodicity of $\cU$.
This amounts to determine the possible permutations $\sigma \in \mathfrak{S}_{\mu}$ such that
\begin{equation}\label{eq:matching}
\frac{aj}{d} - \frac{b\sigma(j)}{d} \in \ZZ , \qquad \forall j = 1, \dots, \bar\mu.
\end{equation}
Imposing \eqref{eq:matching} for $j=i$ and $j=1$ implies $ ( \sigma(i)-i\sigma(1) ) \frac{b}{d} \in \ZZ$.
Hence, for all $i=2,\dots,\bar\mu$, $d$ is a divisor of $\sigma(i)-i\sigma(1)$ since $b$ and $d$ are coprime integers.
Suppose that, for some $i$, $\sigma(i)-i\sigma(1)=0$.
One has $1 \le \sigma(j) \le \bar\mu$ for all $j$, hence one gets $|\sigma(i+1)-(i+1)\sigma(1)| = |\sigma(i+1)-\sigma(i)-\sigma(1)| < 2\bar\mu \le d$.
Therefore $d$ cannot divide $\sigma(i+1)-(i+1)\sigma(1)$ unless $\sigma(i+1)-(i+1)\sigma(1)=0$.
Considering now the case $i=2$, one has $|\sigma(2)-2\sigma(1)| < 2\bar\mu \le d$, hence $d$ cannot divide $\sigma(2)-2\sigma(1)$ unless $\sigma(2)-2\sigma(1)=0$.
Finally, taking $i=\bar\mu$, the recurrence leads to $\sigma(\bar\mu)-\bar\mu\sigma(1)=0$. Since $\sigma(\bar\mu) \le \bar\mu$, one has $\sigma(1)=1$.
It follows that the only possible choice of $\sigma$ is $\sigma(j)=j$ for all $j = 1, \dots, \bar\mu$, i.e., the identity.

The only consistent matching condition is thus simply $\frac{a-b}{d}\in\ZZ$,
complementing the relations $d=d'$ and $\bar\mu=\bar\mu'$ found at step~1.

Let us now further our analysis of these abelianity conditions.
The results obtained in step~1 show that one has to deal with two cases, depending on the relative positions of $\mu$ and $d-\mu$ and of $\mu'$ and $d-\mu'$ on the one hand, and on the signs of the two integers $m$ and $n$ on the other hand, namely:
\begin{align}
\label{eq:caseI}
\text{case (I)} &\quad m = d\bar s \pm \bar\mu \qquad \text{and} \qquad n = d\bar s' \mp \bar\mu
 (\bar s,\bar s' \in \ZZ) \quad \Rightarrow \quad \frac{m+n}{d} \in \ZZ, \\
\label{eq:caseII}
\text{case (II)} &\quad m = d\bar s \pm \bar\mu \qquad \text{and} \qquad n = d\bar s' \pm \bar\mu
 (\bar s,\bar s' \in \ZZ) \quad \Rightarrow \quad \frac{m-n}{d} \in \ZZ.
\end{align}
From the matching condition, one can set $a=\bar ad+\gamma$ and $b=\bar bd+\gamma$ where $\bar a,\bar b \in \ZZ$, $0 < \gamma < d$ and $(\gamma,d)$ are coprime integers (the last two conditions ensure that $(a,d)$ and $(b,d)$ are pairs of coprime integers). The surface condition $\lambda+\lambda^*=1$ then takes the form
\begin{equation}\label{eq:surf2}
\lambda+\lambda^* = \bar a m + \bar b n + \frac{\gamma}{d} (m+n) = 1.
\end{equation}
When $m+n = 0$, equation \eqref{eq:surf2} reads $(\bar a- \bar b)m=1$, which cannot be satisfied except in the very particular cases
$m=-n= \pm 1$. These cases are dealt with by Proposition~\ref{prop:Scrit} and correspond to an extended centrality condition.

We now consider $m+n \neq 0$. The above hypotheses imply that $d$ is a divisor of $m+n$, irrespective to case~(I) or~(II).

In case (II), $d$ is a divisor of both $m+n$ and $m-n$, it is therefore a divisor of $2m$, hence $2\bar\mu/d$ is an integer. The upper bound $\bar\mu \le d$ then implies $\bar\mu = d/2$, which shows that $d$ should be even in that case. Then, $n = d\bar s' \pm d/2$ can be rewritten as $n = d(\bar s'\pm1) \mp d/2$, and the case~(II) appears as a subcase of case~(I).

Thus, we have proved that when $\lambda, \lambda^*$ are not integers, the abelianity property is equivalent to condition 2. This ends the proof of Lemma \ref{lem:abel}.

\subsection{Proof of Lemma \ref{lem:abel2}\label{proof:abel2}}
 Let us now work out the condition~(2) given in Lemma~\ref{lem:abel}.

Let $g=\gcd(m,n)$, $m=\bar m g$, $n=\bar n g$, where $\bar m,\bar n$ are coprime numbers, and similarly, set $u = \gcd(g,d)$, $d = \bar d u$ with $g = \bar g u$. Equation \eqref{eq:surf2} then writes
\begin{equation}\label{eq:surf3}
\lambda+\lambda^* = \bar g \left( u\bar a \bar m + u\bar b \bar n + \frac{\gamma}{\bar d} (\bar m+\bar n) \right) = 1,
\end{equation}
which implies $\bar g=1$ (note that $\bar d$ is a divisor of $\bar m + \bar n$ since $d$ is a divisor of $m+n$ and $\bar d$, $\bar g$ are coprime numbers).
Hence, recalling~\eqref{eq:caseI}, equation \eqref{eq:surf3} takes the form
$g(\bar a \bar m + \bar b \bar n) + \gamma (\bar s+\bar s') \allowbreak = 1$,
showing that $g$ and $\bar s+\bar s'$ shall be coprime numbers with B\'ezout coefficients $\gamma' = \bar a \bar m + \bar b \bar n$ and $\gamma$.

The equation $\gamma' g + \gamma(\bar s+\bar s')=1$ is therefore a constraint equation.
Different cases may arise:
\begin{enumerate}\itemsep=0pt
\item[$i)$] The integers $g=\gcd(m,n)$ and $\bar s+\bar s'=(m+n)/d$ are \emph{not} coprime numbers. There is no solution, in other words no ``cross-cancellation'' can occur.
\item[$ii)$] The integers $g=\gcd(m,n)$ and $\bar s+\bar s'=(m+n)/d$ are coprime numbers.
\end{enumerate}
In that case, note that one can write in general
\begin{equation*}
m+n = g \prod_i d_i \prod_j d'_j,
\end{equation*}
where the $d_i$'s and $d'_j$'s are prime integers, the $d_i$'s are divisors of~$g$, and the $d'_j$'s are coprimes with $g$. Then the admissible divisors of $m+n$, i.e., the divisors $d$ such that $g=\gcd(m,n)$ and $(m+n)/d$ are coprime numbers, are of the form
\begin{equation*}
d = g \prod_i d_i \prod_{j'} d'_{j'},
\end{equation*}
where the set of indices $j'$ is some subset of the set of indices~$j$.

The conditions on $\gamma$ and $d$ given above must now be examined on a case-by-case basis:
\begin{itemize}\itemsep=0pt
\item[--] Either there \emph{does not} exist B\'ezout coefficients $\gamma$ satisfying $0 < \gamma < d$ with $\gamma$,$d$ coprime numbers. The same negative conclusion holds.

\item[--] Or such B\'ezout coefficients exist, and some ``cross-cancellations'' occur.
\end{itemize}
Consider some $\gamma$ with the required properties and denote by $\gamma'$ the other B\'ezout coefficient.
If $(\ell,\ell')$ are the B\'ezout coefficients of $\bar m$ and $\bar n$, $\ell \bar m + \ell' \bar n = 1$, the solution of $\gamma' = \bar a \bar m + \bar b \bar n$ is then given by
\begin{equation*}
\bar a = \gamma' \ell + k \bar n \qquad \text{and} \qquad \bar b = \gamma' \ell' - k \bar m ,
\end{equation*}
with $k\in\ZZ$. Finally, one obtains
\begin{equation*}
\frac{\lambda}{m} = \gamma' \ell + \frac{\gamma}{d} + k \frac{n}{g} \qquad \text{and} \qquad
\frac{\lambda^*}{n} = \gamma' \ell' + \frac{\gamma}{d} - k \frac{m}{g} .
\end{equation*}
This ends the proof of Lemma~\ref{lem:abel2}.

Remark that when $m$ and $n$ are coprimes ($g=1$), the relation $\gamma' g + \gamma(\bar s+\bar s')=1$ always admits a solution. In fact, it allows to eliminate $\gamma'$ and simplifies the expression of~$\lambda$.

\section{Technical proofs\label{sect:proofs}}

Having characterized the abelianity lines, we are now in a position to prove the results presented in Section~\ref{sect:main}.
\subsection{Proof of Theorem \ref{thm:line}}

We first consider the abelianity lines of type (a) in Theorem~\ref{thm:line} and the condition~(1) of Lemma~\ref{lem:abel}.
 If equation~\eqref{eq:thmb} holds, this is equivalent to impose $\lambda\in\ZZ$. Hence, thanks to Proposition~\ref{prop:ligne}, the intersection again defines a line of abelianity.

To show that in case (a) a line of abelianity of $\mathscr{S}_{m,n}$ is not necessarily
a line of abelianity of $\mathscr{S}_{m',n'}$, it is sufficient to exhibit an example.
Indeed, if one considers the intersection of the critical surfaces $\mathscr{S}_{3,6}$ and
$\mathscr{S}_{2,5}$, we don't get a line of abelianity of type~(b). However, it defines a line of abelianity of type (a) on $\mathscr{S}_{3,6}$, but not on $\mathscr{S}_{2,5}$.

It remains to analyse the case (b) in Theorem~\ref{thm:line}. Recalling the expressions \eqref{eq:lambda-inter}, one sees that
\eqref{eq:thma} is exactly the condition $\frac{a-b}d\in\ZZ$ of Lemma~\ref{lem:abel}, with in addition the condition $(m+n)(m'+n')\neq0$.
We set $\delta = \gcd(m+n,m'+n')$, $m+n=\delta u$, and $m'+n'=\delta v$. Equation~\eqref{eq:thma} leads to $u(1 + \alpha n') = v(1 + \alpha n)$, hence $1+\alpha n'=\xi v$, $1+\alpha n=\xi u$ where $\xi\in\ZZ$, since $u,v$ are coprime integers.
Similarly, $u(1 - \alpha m') = v(1 - \alpha m)$, hence $1-\alpha m'=\xi' v$, $1-\alpha m=\xi' u$ where $\xi'\in\ZZ$.
It follows that $\frac{m'-m}{m'n-mn'}=\frac{\xi'}{\delta}$ and $\frac{n-n'}{m'n-mn'}=\frac{\xi}{\delta}$.
Therefore, $d$ is necessarily a divisor of $\delta$, hence a divisor of $m+n$ and $m'+n'$. This implies that the second abelianity condition $\bar\mu=\bar\mu'$ is realized, see~\eqref{eq:caseI}.
Then, type~(b) in Theorem~\ref{thm:line} implies condition~(2) of Lemma~\ref{lem:abel}.

To consider the reciprocal implication, one has to deal with the specific situations where $m+n=0$
or $m'+n'=0$. But it has been shown (see after equation \eqref{eq:surf2}) that it can occur only when $m=-n=\pm1$ or $m'=-n'=\pm1$. This leads to the types (c) and (c$'$) in Theorem~\ref{thm:line}.
Then, we conclude that condition~(2) of Lemma~\ref{lem:abel}
is equivalent to cases (b), (c) and~(c$'$) of Theorem~\ref{thm:line}.

Obviously, since the condition corresponding to case (b) is symmetric in the ex\-chan\-ge
$(m,n)\allowbreak \leftrightarrow(m',n')$, the line of abelianity for $\mathscr{S}_{m,n}$ is also a line of
abelianity for $\mathscr{S}_{m',n'}$. Abelianity for the surface $\mathscr{S}_{1,-1}$ is automatic, see Proposition~\ref{prop:Scrit}.

\subsection{Proof of Theorem \ref{thm:surf}}

We first show that any line of abelianity can be constructed as an intersection. A line of abelianity is characterized by a rational value of $\lambda$ given the surface $\mathscr{S}_{m,n}$. But any rational $\lambda$ can be parametrized by formula~\eqref{eq:lambda-inter} for suitable values of $m'$ and $n'$. Indeed, if we choose $m'=(a+1)m+d$ and $n'=(a+1)n$, we get $a/d$ for the irreducible fraction of $\lambda/m$. It follows that an abelianity line can be identified with the intersection of two surfaces $\mathscr{S}_{m,n}$ and $\mathscr{S}_{m',n'}$.

Let $\mathscr{S}_{m,n}$ and $\mathscr{S}_{m',n'}$ be two intersecting surfaces, hence the parameters $p$ and $p^*$ are given by \eqref{eq:ppstar}. We recall that this equation leads to~\eqref{eq:lambda-inter} through the parametrization~\eqref{eq:paramppstar}.

Case (a) of Theorem~\ref{thm:line} corresponds to $\lambda \in \ZZ$.
Equation~\eqref{eq:lambda-inter} implies that $m'n\lambda+n'm\lambda^*=mn$.
But $\lambda$ and $\lambda^*=1-\lambda$ are coprime integers.
Let $\ell_0$ and $\ell'_0$ be their B\'ezout coefficients such that $\ell_0\lambda+\ell'_0\lambda^*=1$.
A solution for $(m',n')$ is then $m'=\ell_0m$ and $n'=\ell'_0n$.
The general expression for the B\'ezout coefficients of $(\lambda,\lambda^*)$ being $\ell=\ell_0+k\lambda^*$ and $\ell'=\ell'_0-k\lambda$, where $k\in\ZZ$, we obtain a countable number of possible pairs given by $m'=m(\ell_0+k\lambda^*)$ and $n'=n(\ell'_0-k\lambda)$.

Case (b) of Theorem \ref{thm:line} corresponds to the condition $\lambda/m - \lambda^*/n \in \ZZ$.
One sets $\lambda/m=a/d$ and $\lambda^*/n=b/d$, where $a/d$ and $b/d$ are irreducible fractions (see proof of Lemma~\ref{lem:abel}), and one looks for pairs $(m',n')$ such that~\eqref{eq:thma} holds.
Thanks to the change of variables $m'=m+m_0$, $n'=n+n_0$, the expression $\lambda/m=a/d$ and $\lambda^*/n=b/d$ lead to
$dm_0=b(m_0n-n_0m)$ and $dn_0=-a(m_0n-n_0m)$, hence $am_0+bn_0=0$, i.e., $m_0=-bu$ and $n_0=au$ where $u\in\QQ$.
We get a countable number of possible pairs $(m',n')$. They are given by $m'=m-bu$ and $n'=n+au$ where $u \in \ZZ/\gcd(a,b)$ since we are looking for integer solutions.

Cases (c) or (c$'$) of Theorem~\ref{thm:line} correspond here to the same discussion. The intersection of~$\mathscr{S}_{1,-1}$ with $\mathscr{S}_{m,n}$ leads to the following values
 \begin{equation*}
-p^{\frac12}=q^{-N\lambda} \qquad\text{with}\quad\lambda=\frac{n+1}{m+n}\qquad\text{and}\qquad c=-N.
\end{equation*}
We consider the intersection of the surface $\mathscr{S}_{m,n}$ with the surface $\mathscr{S}_{m',n'}$, where we choose $m'=m-u(m-1)$ and $n'=n-u(n+1)$, $u\in\ZZ$. It leads to the values
 \begin{equation*}
-p^{\frac12}=q^{-N\lambda/m} \qquad\text{with}\quad\frac{\lambda}m=\frac{n+1}{m+n}\qquad\text{and}\qquad c=-N.
\end{equation*}
Thus, it defines the same line of abelianity on $\mathscr{S}_{m,n}$. When $u$ varies in $\ZZ$, we get a countable number of surfaces that intersect on this line. Note that $u=0$ corresponds to $\mathscr{S}_{m,n}$, while $u=1$ leads to $\mathscr{S}_{1,-1}$.

\subsection{Proof of Propositions \ref{prop:Scrit}, \ref{prop:S0n} and \ref{prop:Smmm}}

\begin{proof}[Proof of Proposition \ref{prop:Scrit}]
The part (i) has already been proved in \cite{AFR19}. It is thus enough to prove (ii). The intersection of
$\mathscr{S}_{1,-1}$ with a generic surface $\mathscr{S}_{m',n'}$ leads to
\begin{equation*}
\lambda = \frac{m'(n'+1)}{m'+n'} \qquad \text{and} \qquad \lambda^* = \frac{n'(1-m')}{m'+n'}
\end{equation*}
for $(\lambda,\lambda^*)$ associated to the surface $\mathscr{S}_{m',n'}$.

It is easy to check that $\frac{\lambda}{m'}-\frac{\lambda^*}{n'}=1$, so that the first part of condition~(2) in Lemma~\ref{lem:abel} is satisfied for all values of~$m'$ and~$n'$. It implies also that the irreducible fractions corresponding to $\frac{\lambda}{m'}$ and $\frac{\lambda^*}{n'}$ have the same denominator $d$. Then, it remains to show that this denominator $d$ divides $m'+n'$.
Since $\frac{\lambda}{m'} = \frac{n'+1}{m'+n'}$, it is immediate.

Now, choosing $n'=a-1$ and $m'=d-a+1$, we get that $\frac{\lambda}{m'}=\frac{a}{d}$ with $m'+n'=d$. In that case, $\frac{\ln p}{2N\ln q} = \frac{a}{d}$: varying $a$ and $d$ in $\ZZ$, we get any rational number, so that the abelianity lines constructed in this way form a dense subset of the surface $\mathscr{S}_{1,-1}$.
\end{proof}

\begin{proof}[Proof of Proposition~\ref{prop:S0n}] The parametrization~\eqref{eq:ppstar} shows obviously that the surfaces~$\mathscr{S}_{0,n}$ have to be studied specifically. Indeed, from the surface condition~\eqref{eq:surf}, that reads now $\big({-}p^{*\frac{1}{2}}\big)^{n} = q^{-N}$, one obtains immediately
\begin{equation*}
\mathcal{Y}_{0,n}(x) = \frac{\displaystyle
\prod_{\ell=0}^{|n|-1} \mathcal{U}\big(\big({-}p^{*\frac{1}{2}}\big)^{\ell} x\big)}
{\displaystyle \prod_{\ell=1}^{|n|} \mathcal{U}\big(\big({-}p^{*\frac{1}{2}}\big)^{-\ell} x\big)} = 1.
\end{equation*}
Hence, the generators $t^{(k)}_{0,n}(z)$ satisfy an abelian algebra on the surface $\mathscr{S}_{0,n}$. This proves the part (i) of the proposition.

Let us now examine the intersection of the surface $\mathscr{S}_{0,n}$ with a generic surface $\mathscr{S}_{m',n'}$. The equations \eqref{eq:ppstar} imply the following expressions for $\lambda$ and $\lambda^*$ relative to the surface $\mathscr{S}_{m',n'}$, i.e., $-p^{\frac{1}{2}}=q^{-N\lambda/m'}$ and $-p^{*\frac{1}{2}}=q^{-N\lambda^*/n'}$:
\begin{equation*}
\lambda = 1-\frac{n'}{n} \qquad \text{and} \qquad \lambda^* = \frac{n'}{n} .
\end{equation*}
If $n' \in n\ZZ$, then $\lambda$ is an integer and one gets an abelianity line of type (a) in Theorem \ref{thm:line} for any value of $m'$. Setting $n'=kn$ ($k\in\ZZ$), one gets
$\frac{\ln p}{2N\ln q} = \frac{k-1}{m'}$: when $k$ and $m'$ run over $\ZZ$, one obtains any rational
number. Hence, the abelianity lines of $\mathscr{S}_{m',n'}$ identified as intersections with $\mathscr{S}_{0,n}$ form a dense set in the surface $\mathscr{S}_{0,n}$. This proves the part (ii).

Finally, choosing $n'=kn+1$ and $m'=(1-k)n$ with $k\in\ZZ$ shows that $\lambda=1-k-\frac1n\notin\ZZ$
while $\frac{\lambda}{m'}-\frac{\lambda^*}{n'}=\frac1{(1-k)n^2}\notin\ZZ$: the intersection is not a line of abelianity for
the surface $\mathscr{S}_{m',n'}$. Running $k$ over $\ZZ$, we get a countable number of such surfaces.
\end{proof}

\begin{proof}[Proof of Proposition~\ref{prop:Smmm}] We recall that, due to the surface condition, the central charge is fixed to the value $c=-N/m$, the other parameters $q$ and $p$ remaining unconstrained. By contrast to the critical case, since $m \ne 1$, the generator $t_{m,-m}(z)$ \emph{does not} commute with the generators of $\gelpn$, hence there is no extended center.

However, under certain supplementary conditions on $p$, one finds a ``localized'' extended center, i.e.m the generator $t_{m,-m}(z)$ commutes with those of $\gelpn$ on a certain submanifold of the surface $\mathscr{S}_{m,-m}$, see Corollary~3.2 in~\cite{AFR17}.
Such a submanifold will be called ``super-abelianity line'' for obvious reasons. In particular, it is
 also an abelianity line.
Let us now propose a characterization of these special lines.
Consider the exchange function \eqref{eq:exchtL} for $m+n=0$:
\begin{equation}
\prod_{k=1}^m \frac{\mathcal{U}\big(\big({-}p^{*\frac{1}{2}}\big)^{-k} x\big)}{\mathcal{U}\big(\big({-}p^{\frac{1}{2}}\big)^{-k} x\big)} \quad \text{for $m>0$} \qquad \text{and} \qquad \prod_{k=0}^{n-1} \frac{\mathcal{U}\big(\big({-}p^{\frac{1}{2}}\big)^{k} x\big)}{\mathcal{U}\big(\big({-}p^{*\frac{1}{2}}\big)^{k} x\big)} \quad \text{for $n>0$}.\label{eq:ratio}
\end{equation}
The ratio~\eqref{eq:ratio} is equal to 1 if each term indexed by $k$ in the numerator simplifies with the term indexed by $\sigma(k)$ in the denominator where $\sigma \in \mathfrak{S}_m$, up to a power of $q^N$ since the function $\mathcal{U}$ is $q^N$-periodic. We stick to the case $m>0$ (the case $m<0$ runs along similar lines), and we set $\cI_m = \{ 1, \dots, m \}$.
Using the parametrization~\eqref{eq:paramppstar}, one gets
\begin{equation*}
\lambda \big(k-\sigma(k) \big) = k - m\ell(k) \quad \text{where} \quad \ell(k) \in \ZZ.
\end{equation*}
One can deduce
\begin{equation}
m\big(\ell(k+1)-\ell(k)\big) + \lambda\big(1-\sigma(k+1)+\sigma(k)\big)=1\label{eq:(1)}
\end{equation}
with the boundary equation
\begin{equation}
m\ell(1) + \lambda\big(1-\sigma(1)\big)=1.\label{eq:(2)}
\end{equation}
Since no abelianity line of type (b) exists on $\mathscr{S}_{m,-m}$, one can restrict to $\lambda\in\ZZ$.
Equation~\eqref{eq:(2)} implies that $m$ and $\lambda$ have to be coprime integers. Let $(\beta_m,\beta_\lambda)$ be the corresponding B\'ezout coefficients, i.e., $\beta_m m - \beta_\lambda \lambda=1$. Their general expression is given by $\beta_m = \beta_0+\alpha\lambda$ and $\beta_\lambda = \beta'_0+\alpha m$ where $\alpha\in\ZZ$ and $(\beta_0,\beta'_0)$ is the representative such that $1 \le \beta'_0 \le m-1$. It follows then from~\eqref{eq:(1)} and~\eqref{eq:(2)}:
\begin{equation*}
\sigma(k+1)-\sigma(k)-1 = \beta_\lambda(k+1) = \beta'_0 + \alpha_{k+1}m
\end{equation*}
and
\begin{equation*}
\sigma(1)-1 = \beta_\lambda(1) = \beta'_0 + \alpha_{1}m.
\end{equation*}
Hence, one gets
\begin{equation}
\sigma(k) = k(1+\beta'_0) + m\sum_{i=1}^k \alpha_{i}.\label{eq:(3)}
\end{equation}
When $m$ is even, $\beta_\lambda$ has to be odd, hence $\beta'_0$ has to be odd: one generates only even values of~$\sigma(k)$. This leaves us only with odd values for $m$. Now, if $\gcd(m,\beta'_0+1) \ne 1$, all values of~$\sigma(k)$ are multiple of the gcd and one does not span the whole set $\cI_m$. Therefore, $m$ and $\beta'_0+1$ have to be coprime integers.
Finally, the condition $1 \le \beta'_0 \le m-1$ implies that the $\alpha_k$'s can always be chosen such that $\sigma(k) \in \cI_m$ for all $k$. Taking now $k,k' \in \cI_m$ with $k \neq k'$, one gets from~\eqref{eq:(3)}
\begin{gather}
\sigma(k') - \sigma(k) = (k'-k)(1+\beta'_0) + m\sum_{i=k+1}^{k'} \alpha_{i}.\label{eq:(4)}
\end{gather}
Since $1 \le |k'-k| \le m-1$ and $m$, $\beta'_0+1$ are coprime integers, the ratio $(k'-k)(1+\beta'_0)/m$ is never an integer and the r.h.s.\ of \eqref{eq:(4)} cannot vanish for any pair $(k,k')$. It follows that the $\sigma(k)$'s span the set $\cI_m$ when $k$ runs over~$\cI_m$.

It follows that super-abelianity lines are necessarily abelianity lines such as characterized by condition~(a) of Theorem~\ref{thm:line}, together with the further algorithmic conditions of primality established in the previous discussion.

Consider now the particular case of the intersection $\mathscr{S}_{m,-m} \cap \mathscr{S}_{1,1}$ with $m$ odd. For the parameters of $\mathscr{S}_{m,-m}$, one obtains $\lambda=\frac{m+1}{2}\in\ZZ$, which is coprime with $m$, since $2\lambda-m=1$.
The corresponding B\'ezout coefficients $\beta_0$ and $\beta'_0$
in $]0 , m[$ are given by $\beta_0=\frac{m-1}{2}$ and $\beta'_0=m-2$, hence $\beta'_0+1$ and $m$ are coprime integers. Therefore, this abelianity line of $\mathscr{S}_{m,-m}$ is a super-abelianity line.

Now as an intersection on $\mathscr{S}_{1,1}$, we get $\lambda=\frac{m+1}{2m}$ which is not an integer, so that the (possible) abelianity may only match condition~(2) of Lemma~\ref{lem:abel}. But $\lambda/1-\lambda^*/1=\frac{1}{m}$ is not an integer, so that the condition 2 is not fulfilled, and we don't have an abelianity line for $\mathscr{S}_{1,1}$.
\end{proof}

Inspired by this observation, one may now look for general values of $m$, $m'$, $n'$ such that the intersection $\mathscr{S}_{m,-m} \cap \mathscr{S}_{m',n'}$ leads to a super-abelianity line. However, given the expression of $\lambda$ above, this is clearly a purely algorithmic problem which goes beyond the scope of this paper.

Note that two surfaces $\mathscr{S}_{m,-m}$ and $\mathscr{S}_{n,-n}$ never intersect when $m\ne n$ (see Lemma~\ref{lem:det}).

\section{Poisson structures\label{sect:poisson}}

Having explicited the conditions under which the quadratic exchange structures in $\gelpn$ lead to abelian subalgebras, one can define Poisson structures on them. The explicit construction of these Poisson structures follows the standard scheme (see, e.g., \cite{AFR17, AFR19}). More precisely, on the surface $\mathscr{S}_{m,n}$, setting $p^{1-\epsilon}=q^{-\frac{2N\lambda}{m}}$ when one of the conditions on $\lambda$ of Lemma~\ref{lem:abel} is satisfied, one defines a Poisson structure by
\begin{equation*}
\big\{ t_{m,n}^{(k)}(z),t_{m,n}^{(k')}(w) \big\} = \lim_{\epsilon \to 0} \frac{1}{\epsilon} \big( t_{m,n}^{(k)}(z) t_{m,n}^{(k')}(w) - t_{m,n}^{(k')}(w) t_{m,n}^{(k)}(z) \big).
\end{equation*}
We recall that we are eluding in \eqref{eq:exchtt} the subleading terms coming from the singularities in the Riemann--Hilbert splitting. Hence the Poisson bracket we obtain is a purely quadratic one. It corresponds to the leading term of the full Poisson structure that would be obtained from the complete achievement of the Riemann--Hilbert procedure.

\begin{Proposition}
On the line of abelianity, the Poisson structure is given by
\begin{equation*}
\big\{ t_{m,n}^{(k)}(z),t_{m,n}^{(k')}(w) \big\} = f^{(k,k')}(z/w) t_{m,n}^{(k)}(z) t_{m,n}^{(k')}(w),
\end{equation*}
where
\begin{equation*}
f^{(k,k')}(x) = \sum_{i=(1-k)/2}^{(k-1)/2} \sum_{j=(1-k')/2}^{(k'-1)/2} f\big(q^{i-j}x\big).
\end{equation*}
The explicit form of the function $f(x)$ depends on the type of line of abelianity
$($see classification Theorem~{\rm \ref{thm:line}} and notation~\eqref{eq:paramppstar}$)$.

For the type $(a)$ lines, we get
\begin{equation}\label{eq:f-b}
f(x) = -N\lambda(\ln q) x \frac{{\rm d}}{{\rm d}x} \left[ \frac{m}{\ell} \ln\cU_{q^{2N/\ell}}(x) + \frac{n}{\ell^*} \ln\cU_{q^{2N/\ell^*}}(x) \right],
\end{equation}
where $\ell=m/w$, $\ell^*=n/w^*$, and $w=\gcd(\lambda,m)$, $w^*=\gcd(\lambda^*,n)$, choosing for $w$ and $w^*$ the signs of $m$ and $n$ respectively.

For lines of type $(b)$, the function reads
\begin{gather}
f(x) = -N\lambda(\ln q) \frac{m+n}{d} x\frac{{\rm d}}{{\rm d}x} \left[ \rule{0pt}{20pt}\right.
 \left(1+\frac{\mu^2}{mn}\right) \ln\cU_{q^{2N/d}}(x) - \frac{d\mu}{mn} \ln\cU_{q^{2N}}\big(x\big)\nonumber \\
 \left. \hphantom{f(x) =}{} + \frac{d}{mn} \sum_{k=1}^{\mu-1} (k-\mu) \ln\Big(\cU_{q^{2N}}\big(\big({-}p^{\frac{1}{2}}\big)^{k}x\big)
\cU_{q^{2N}}\big(\big({-}p^{\frac{1}{2}}\big)^{-k}x\big) \Big) \right],\label{eq:f-a}
\end{gather}
where $d>0$ is the divisor of $m+n$ characterizing the abelianity line, defined in Lemma~{\rm \ref{lem:abel}}, condition~$(2)$, and $\mu$ is the remainder of the Euclidean division of~$m$ by~$d$ such that $0 < \mu < d$.

Here, we have introduced the function
\begin{gather}\label{eq:defUa}
\cU_a(x) = q^{\frac2N-2} \frac{\theta_{a}\big(q^2z^2\big) \theta_{a}\big(q^2z^{-2}\big)} {\theta_{a}\big(z^2\big)\theta_{a}\big(z^{-2}\big)} .
\end{gather}
\end{Proposition}

\begin{proof} It follows from the results of \cite{AFR19} that the Poisson structure is generically given by
the function $f(x)$ itself corresponding to the Poisson structure related to the abelian DVA-like subalgebra in $\gelpn$ generated by $t_{m,n}^{(1)}(z)$. Hence, it is sufficient to study this latter case,
according to the discussion on the exchange structure function $\mathcal{Y}_{m,n}(x)$, see end of Section~\ref{sect:Wpq}.

Given the expression of the exchange function $\mathcal{Y}_{m,n}(x)$, see~\eqref{eq:funcY}, the general structure of~$f(x)$ reads
\begin{equation*}
f(x) = \frac{{\rm d}}{{\rm d}\epsilon} \left( -\sum_{k=1}^{|m|-1} \ln y_k(x) + \sum_{k=1}^{|n|-1} \ln y^*_k(x) - \ln y_b(x) \right),
\end{equation*}
where
\begin{equation*}
y_k(x) = \frac{\mathcal{U}\big(\big({-}p^{\frac{1}{2}}\big)^{-k} x\big)}{\mathcal{U}\big(\big({-}p^{\frac{1}{2}}\big)^{k} x\big)} , \qquad y^*_k(x) = y_k(x)\vert_{p \to p^*} , \qquad
y_b(x) = \frac{\mathcal{U}\big(\big({-}p^{\frac{1}{2}}\big)^{-|m|} x\big)}{\mathcal{U}\big(\big({-}p^{*\frac{1}{2}}\big)^{-|n|} x\big)}.
\end{equation*}
In each case, the explicit form of the Poisson structure is given by a direct (but lengthy) calculation of the derivative, using the definition of the short Jacobi~$\theta$ function as absolute convergent products for $|q|<1$.

\textbf{Case (a): $\boldsymbol{\lambda \in \ZZ}$.} The function $f(x)$ is given by
\begin{equation*}
f(x) = -2N\lambda(\ln q) \big( 2\mathcal{I}(x) - \mathcal{I}(qx) - \mathcal{I}\big(q^{-1}x\big)\big),
\end{equation*}
where
\begin{gather}
\mathcal{I}(x) = \frac{m}{\ell} \left( \sum_{s=0}^{\infty} \frac{x^2q^{2Ns/\ell}}{1-x^2q^{2Ns/\ell}} - \sum_{s=1}^{\infty} \frac{x^{-2}q^{2Ns/\ell}}{1-x^{-2}q^{2Ns/\ell}} \right) \nonumber \\
\hphantom{\mathcal{I}(x) =}{} + \frac{n}{\ell^*} \left( \sum_{s=0}^{\infty} \frac{x^2q^{2Ns/\ell^*}}{1-x^2q^{2Ns/\ell^*}} - \sum_{s=1}^{\infty} \frac{x^{-2}q^{2Ns/\ell^*}}{1-x^{-2}q^{2Ns/\ell^*}} \right).\label{eq:poissonb}
\end{gather}
From the definition of $w$ and $w^*$, $\ell$ and $\ell^*$ are identified with the denominators of the reduced form of the rationals $\lambda/m$ (resp.~$\lambda^*/n$), themselves identified up to $2N$ with the ratio $\ln p/\ln q$ and $\ln p^*/\ln q$.

\textbf{Case (b): $\boldsymbol{\lambda/m-\lambda^*/n \in \ZZ}$.} The function $f(x)$ is given by
\begin{equation*}
f(x) = -2N\lambda(\ln q) \frac{m+n}{d} \big( 2\mathcal{I}(x) - \mathcal{I}(qx) - \mathcal{I}\big(q^{-1}x\big) \big),
\end{equation*}
where
\begin{gather}
\mathcal{I}(x) = \left(1+\frac{\mu^2}{mn}\right) \left( \sum_{s=0}^{\infty} \frac{x^2q^{2Ns/d}}{1-x^2q^{2Ns/d}} - \sum_{s=1}^{\infty} \frac{x^{-2}q^{2Ns/d}}{1-x^{-2}q^{2Ns/d}} \right) \nonumber \\
\hphantom{\mathcal{I}(x) =}{} + \frac{d\mu}{mn} \left( \sum_{s=0}^{\infty} \frac{x^2q^{2Ns}}{1-x^2q^{2Ns}} - \sum_{s=1}^{\infty} \frac{x^{-2}q^{2Ns}}{1-x^{-2}q^{2Ns}} \right)
+ \frac{d}{mn} \sum_{k=0}^{\mu-1} (k-\mu) \left( \sum_{s=0}^{\infty} \frac{x^2p^{k}q^{2Ns}}{1-x^2p^{k}q^{2Ns}} \right. \nonumber \\
 \left.\hphantom{\mathcal{I}(x) =}{} + \sum_{s=1}^{\infty} \frac{x^2p^{-k}q^{2Ns}}{1-x^2p^{-k}q^{2Ns}} - \sum_{s=0}^{\infty} \frac{x^{-2}p^{k}q^{2Ns}}{1-x^{-2}p^{k}q^{2Ns}} - \sum_{s=1}^{\infty} \frac{x^{-2}p^{-k}q^{2Ns}}{1-x^{-2}p^{-k}q^{2Ns}} \right) .\label{eq:poissona}
\end{gather}

It can be verified that the formulae \eqref{eq:poissonb} and \eqref{eq:poissona} remain valid when $|m|=1$ or $|n|=1$. Case (b) then only occurs if $n$ or $m$ are such that $\mu=1$ or $\mu=d-1$. In case~(a), note that $w=m$ when $|m|=1$ and $w^*=n$ when $|n|=1$.

Finally, the function $f(x)$ can be rewritten in a more compact form as the logarithmic derivative with respect to $x$ of the function $\cU_a$ with parameters $a=q^{2N/d}$, $a=q^{2N/\ell}$ or $a=q^{2N/\ell^*}$.
Indeed, considering the short Jacobi $\theta$ function with elliptic nome $a$, one has
\begin{equation*}
-x \frac{{\rm d}}{{\rm d}x} \ln\theta_a(x) = \sum_{s=0}^{\infty} \frac{xa^s}{1-xa^s} - \sum_{s=1}^{\infty} \frac{x^{-1}a^s}{1-x^{-1}a^s} .
\end{equation*}
Introducing the function $\cU_a(x)$ defined in \eqref{eq:defUa}, one gets the expressions~\eqref{eq:f-a} and~\eqref{eq:f-b}.
\end{proof}

It is important to point out the overlap of the abelianity conditions between the two cases. It occurs when the rest $\mu$ becomes zero in case~(b), and when the reduced denominators of~$\lambda/m$ and~$\lambda^*/n$ coincide in case~(a), i.e., $\ell=\ell^*=d$. The structure functions given in both formulae coincide as it should be.

\section{Conclusion}\label{sect:conclu}

The results we have obtained on abelianity lines, their characterization as intersection of critical surfaces, and their associated Poisson structures, suggest some further lines of investigation. Let us propose a few such directions.

First, since we are dealing with intersections of critical surfaces, we have several types of $\cW^{(m,n)}_{pqc}(N)$ algebras defined simultaneously on these lines. Since the abelianity condition is not always symmetric, it is clear that these algebras cannot be always identical. However, we have proved that each intersection corresponds to a countable number of surfaces, and thus a countable number of $\cW^{(m,n)}_{pqc}(N)$ algebras. Then, it is likely that some of them may coincide, and an analysis on the number of truly different algebras on each line is certainly worth completing. In the same way, when the intersection defines an abelianity line for both surfaces, the corresponding algebras are obviously isomorphic, but it would be interesting to look at the realizations in~$\gelpn$, and see if the generators are indeed identical.

We have derived several sets of Poisson bracket structures, characterized as surface-dependent linear combinations of solely line-dependent elliptic functions. As is always the case~\cite{AFR17,AFR19}, this abstract derivation provides only the leading spin terms of the $q$-$W_N$ Poisson algebra, and their consistent quantizations along the critical surfaces. Only an explicit realization, e.g., by vertex operators will provide the lower spin and central extension terms (see also~\cite{AFRScentral} for a systematic resolution of coboundary conditions). It must be emphasized indeed that a number of realizations of DVA by $q$-bosonized vertex operators, derived from $U_q(\mathfrak{sl}(2))$ generators, have been proposed. The earliest ones were constructed as soon as the DVA algebra itself~\cite{SKAO}. Very recently some new constructions were achieved~\cite{BerGo}. The question here is to find the suitable deformation of free boson algebra yielding as leading order of the exchange structure our abstract DVA and more generally $q$-$W_N$ algebras. It is amusing to note that a realization of a distinct DVA algebra, conjectured in~\cite{JiShi} was given directly~\cite{Shi} in terms of VO of the elliptic quantum algebra $\gelp$ for some particular values of $\ln p/\ln q$. A curious connection thus arises again between elliptic quantum algebras and DVA.

The meaning of the integer conditions in Theorem \ref{thm:line} remains to be investigated. The equation of the critical surface is identified, in terms of coordinates $\ln p$, $\ln p^*$, $2N\ln q$, with a~condition of linearity with the directing vector $(m, n, 1)$ defining the orthogonal direction to the critical plane. This vector must be in fact understood as a projective object with suitable integer conditions; in particular it belongs to the subspace (or ``manifold chart'' if we were not dealing with integers) characterized by a non zero third component~$x_3$. The conditions in Lemma~\ref{lem:det} mean that two surfaces intersect iff the vector product of their respective directing vectors has three non-zero components. In particular, it belongs to the consistent chart $x_3 \neq 0$ of the projective 3d vector space. It remains to see whether a geometric interpretation along these lines may then exist for the abelianity conditions.

Still in the light of Theorem~\ref{thm:line}, it would be interesting to understand the algebraic structure occurring on intersection of surfaces, when the abelianity conditions are not fulfilled. Obviously, the Poisson structure introduced in the abelian case cannot be reproduced outside abelianity. However, a more general structure may arise that would generalize the notion of symplectic structure.
One could think for instance of a trace brackets structure, or a Poisson vertex algebra.

Finally, other deformations of $W$-algebras have been considered in the literature.
One such example was proposed in~\cite{KoKo} in relation to the algebra $U_{q,p}(\wh{\mathfrak{sl}}_N)$.
We addressed in \cite{AFRdyn} the related question of extended center\footnote{Remind that the existence of such an extended center was at the core of the original approach to $q$-deformed $W$-algebras~\cite{FR}.}
at the critical value of the central charge in the  algebra~$\mathcal{B}_{q,\lambda}\big(\widehat{\mathfrak{gl}}(2)_c\big)$.
The former algebra can be viewed as the tensor product of the latter by an Heisenberg algebra~\cite{KojiKo}.
In both cases, the resulting algebra was not dynamical.
It was then argued in~\cite{AFRdyn} that the $W$-algebras build in the $\mathcal{A}_{q,p}\big(\widehat{\mathfrak{gl}}(2)_{c}\big)$
and $\mathcal{B}_{q,\lambda}\big(\widehat{\mathfrak{gl}}(2)_c\big)$ cases should be related through the vertex-IRF correspondence.

Another example consists of double deformation of $W$-algebras.
The underlying algebraic structures are based on quiver algebras~\cite{KP} or toroidal algebras~\cite{Nieri}
and have been shown~\cite{KP2} to generalize the construction proposed in~\cite{FR98}.
It would be interesting to investigate how critical surfaces and abelianity conditions arise in the context of double deformation.

\subsection*{Acknowledgements}
Part of this work was done during the visit of J.A.\ to LAPTh, with financial support from the USMB grant AAP ASI-32.
We also wish to thank the referees for fruitful comments.

\pdfbookmark[1]{References}{ref}
\LastPageEnding

\end{document}